\documentclass[]{iopart}
\usepackage{graphicx}
\usepackage{bm}
\usepackage{amssymb}

\newcommand{\Exp}[1]{\,\mathrm{e}^{\mbox{\footnotesize$#1$}}}
\newcommand{\I}{\mathrm{i}}
\newcommand{\floor}[1]{\left\lfloor #1 \right\rfloor}
\newcommand{\ket}[1]{|#1\rangle}
\newcommand{\bra}[1]{\langle#1|}
\newcommand{\expectn}[1]{\langle#1\rangle}
\def \del{\partial}
\newenvironment{eqnArray}{\begin{equation}\begin{array}[b]{rcl}}
{\end{array}\end{equation}}
\newcommand{\openone}{\mathbb{I}}
\def \algc{{\cal C}}
\def \cK{{\cal K}}
\def \cH{{\cal H}}
\def \cD{{\cal D}}
\def \cA{{\cal A}}
\def \alga{{\cal A}}
\def \algqbit{{\cal A}_{qbit}}
\def \algbofh{{\cal B}({\cal H})}

\def \bbr{{\mathbb R}}
\def \bbc{{\mathbb C}}
\def \bbz{{\mathbb Z}}
\def \bbn{{\mathbb N}}

\newtheorem{theorem}{Theorem}[section]
\newtheorem{lemma}[theorem]{Lemma}
\newtheorem{proposition}[theorem]{Proposition}
\newtheorem{corollary}[theorem]{Corollary}

\newenvironment{proof}[1][Proof:]{\begin{trivlist}
\item[\hskip \labelsep {\bfseries #1}]}{\end{trivlist}}

\newenvironment{definition}[1][Definition]{\begin{trivlist}
\item[\hskip \labelsep {\bfseries #1}]}{\end{trivlist}}

\newcommand{\qed}{\nobreak \ifvmode \relax \else
      \ifdim\lastskip<1.5em \hskip-\lastskip
      \hskip1.5em plus0em minus0.5em \fi \nobreak
      \vrule height0.75em width0.5em depth0.25em\fi}

\begin{document}
\title[Qubit subalgebra and tensor product in angular momentum system]%
{Qubit subalgebra and tensor product in Weyl algebra of angular momentum system}
\author{Jun Suzuki}

\address{%
Graduate School of Information Systems, %
The University of Electro-Communications, %
1-5-1 Chofugaoka, Chofu, Tokyo 182-8585, Japan}

\ead{junsuzuki@is.uec.ac.jp}

\begin{abstract}
We analyze Weyl algebra of quantum angular momentum system and 
construct qubit subalgebra out of it. We show that the commutant of this qubit subalgebra 
is isomorphic to the original algebra and prove the tensor product structure between 
qubit subalgebra and its commutant. This construction can be iterated to 
construct arbitrary number of qubit subalgebras from a single quantum system. 
We show a simple experimental realization of this proposed scheme using 
orbital angular momentum of single photons. 
We briefly discuss about construction of qudit subalgbra and 
generalization to other infinite dimensional systems. 
\end{abstract}

\pacs{03.65.-w,03.65.Fd}  
\ams{08A30}               

\renewcommand{\submitto}[1]{\vspace{28pt plus 10pt minus 18pt}
     \noindent{\small\textrm{Posted on the arXiv on #1}}}
\submitto{4 Dec 2014}

\section{Introduction}\label{sec1}
This paper addresses a problem of constructing qubit subsystems 
out of a single quantum system of infinite dimensional Hilbert spaces. 
This problem is important not just from theoretical studies but 
from a practical point of view. This is partly because real physical 
systems are often described by infinite dimensional Hilbert space. 
The other reason is, as we shall show in this paper, 
that infinite dimensional systems might be 
more useful than finite dimensional systems with a fixed dimension 
when we encode many qubits in a single physical system. 

There have been many activities on the subject in past, see for example 
Refs.~\cite{GKP, BGS,RCMMG,MLGWRN}. 
In previous studies, however, constructed ``qubits'' are only resemble to 
a true qubit system which will be defined in this paper. 
Another important point which will be addressed in this paper 
is that constructed qubit subsystem has a tensor product structure to the rest 
rather than a direct sum structure. The former construction scheme is referred 
to as a subsystem encoding and the latter is known as a subspace encoding in literature. 
This point becomes crucial when one wishes to construct 
several qubit subsystems out of a single physical system. 
The subspace encoding is in general suffering from a leakage problem when is exposed to noise. 
The subsystem encoding is thus superior encoding scheme from practical point of view. 

The subsystem encoding has been discussed for finite dimensional systems in 
the general setting and there are many interesting results known on 
the subject \cite{Z01,VKL, ZLL,BKOSV,GBRWKL,TBKN}. 
For infinite dimensional systems, however, little results are known based on physical models 
\footnote{We have noticed the reference \cite{LS04} which attempts to generalize 
the result of \cite{Z01,ZLL} to infinite dimensional systems. Ref.~\cite{LS04} 
focuses on the structure of Hilbert space itself rather than algebraic aspects.}. 
We note that this problem is a study of subalgebra from mathematical point of view 
and there are general results known in $*$-algebra. 
Yet, these results are usually abstract and are not connected to any physical models. 

In the previous publication \cite{RKSE10}, we proposed an encoding scheme of 
constructing many qubits out of a single rotor system 
and gave possible realizations of our scheme using orbital angular momentum of single photons. 
We showed a single quantum system can potentially 
perform arbitrary quantum information processing protocols. 
As an application, we also showed quantum error correction is possible 
for a single photon system \cite{apj}. 
All results are, however, obtained based on physical intuition, 
and mathematical justification of our results has not been provided so far. 
It is our main objective to discuss details of our proposed scheme. 

We start from a von Neumann algebra $\alga$ on an infinite dimensional 
Hilbert space $\cH$ \cite{araki,Tbook,UOHbook}. 
We construct a $*$-subalgebra $\alga_{qbit}\subset \alga$ 
which is $*$-isomorphic to a matrix algebra on $\mathbb{C}$. 
We then seek a condition which purports to introduce a tensor product structure 
for the subalgebra $\alga_{qbit}$. 
We note the necessary and sufficient condition 
is already derived for finite dimensional Hilbert spaces \cite{Z01,ZLL}.  
To solve this problem, we construct a qubit subalgebra from 
the Weyl algebra of quantum angular momentum system 
and examine properties of this subalgebra in detail. 
For infinite dimensional systems, the main ingredient to analyze algebraic properties 
is the commutant of $\alga_{qbit}$ defined by
\begin{equation}
\algqbit':=\{ b\in\alga\, |\, [b,a]=0,\forall a\in\alga_{qbit}\}. 
\end{equation}
We shall show that there exists a $*$-isomorphism 
from $\algqbit\otimes\algqbit'$ to $\alga$ to 
conclude a tensor product structure $\alga\cong\algqbit\otimes\algqbit'$. 

Upon proving the above statement, we note that C$^*$-algebra is 
too general and rather the original algebra needs to be a von Neumann algebra. 
This is because a tensor product for a von Neumann algebra is unique whereas 
it is not for C$^*$-algebra. Thus, it is important to study a von Neumann algebra 
for infinite dimensional systems for our purpose. 

As the second result, we show that the commutatnt is 
$*$-isomorphic to the original algebra $\alga$, i.e., $\algqbit'\cong\cA$ and 
this provides a simple encoding scheme to construct arbitrary number of qubit subalgebras 
out of a single physical system as 
\begin{equation} \label{ntensor}
{\cal A}\cong {\cal A}_{qbit} \otimes{\cal A}\cong
 \underbrace{{\cal A}_{qbit} \otimes{\cal A}_{qbit} \otimes\cdots\otimes{\cal A}_{qbit}}_n \otimes {\cal A}.
\end{equation}

This paper is organized as follows. 
Section \ref{sec2} provides definitions of algebras under consideration and 
basic properties of them. 
Section \ref{sec3} analyzes the qubit subalgebra of quantum angular momentum algebra. 
Section \ref{sec4} proves our main claim, a tensor product structure for 
qubit subalgebra and its commutant. 
Section \ref{implementation} gives a short account on how to implement our qubit subalgebra 
with orbital angular momentum degrees of freedom of single photons. 
Section \ref{sec6} discusses a possible extension of our proposal briefly and 
the last section summarizes our result. 

\section{Preliminaries}\label{sec2}
\subsection{Qubit algebra}\label{qubitalg}
In this paper, we define a qubit algebra as a $2\times2$ matrix algebra on $\mathbb{C}$. 
\begin{definition}(Qubit algebra)
\begin{equation}  
M_2(\mathbb{C}):=\mathrm{span}_{\bbc}\{I, \sigma_1,\sigma_2,\sigma_3\} 
\end{equation}
where $I$ is the identity matrix and 
$\sigma_{j}$ ($j=1,2,3$) are usual Pauli spin operators satisfying 
the following conditions for $j,k=1,2,3$: 
\begin{eqnArray}\nonumber
\sigma_{j}&=&\sigma_{j}^{*},\\ 
\sigma_j \sigma_k&=&\delta_{jk} I+\I \sum_{\ell=1,2,3}\epsilon_{jk\ell}\ \sigma_{\ell}.
\end{eqnArray}
Here $\sigma^*$ denotes hermite conjugation of $\sigma$, 
$\delta_{jk}$ is the Kronecker delta, and $\epsilon_{jk\ell}$ is the totally antisymmetric tensor. 
\end{definition}
The notation $\sigma_{\pm}=(\sigma_1\pm \I\sigma_2)/2$ is also used for convenience. 
The standard matrix (irreducible) representations for $\sigma_j$ on $\mathbb{C}$ are
\begin{equation} \label{pauli}
\sigma_1\!=\!\left(\begin{array}{cc} ~0 & ~1 \\ ~1 & ~0\end{array}\right),\ 
\sigma_2\!=\!\left(\begin{array}{cc} ~0 & -i \\ ~i & ~0\end{array}\right),\ 
\sigma_3\!=\!\left(\begin{array}{cc} ~1 & ~0 \\ ~0 & -1\end{array}\right).
\end{equation}
We call a subalgebra $\algqbit$ is a qubit algebra when $\algqbit$ is $*$-isomorphic to the algebra $M_2(\mathbb{C})$ 
throughout our discussion. 

\subsection{Angular momentum algebra}

Quantum angular momentum operator $L$ about a given axis is described by 
the differential operator with respect to angle $\theta$ as 
$L=-\I \del/\del\theta$ whose domain is a dense subset of 
the Hilbert space of all square integrable functions on a torus: 
\begin{equation}
L^2([0,2\pi)):=\left\{ f(\theta)\,|\, f:[0,2\pi)\to\mathbb{C},\, f(\theta+2\pi)=f(\theta)\right\}.
\end{equation}
The inner product for $\cH$ is
$\expectn{f|g}=\int_0^{2\pi}\! d\theta\ \overline{f(\theta)}g(\theta)$, 
$f,g\in L^2([0,2\pi))$ with a bar denoting complex conjugation. 
It is known that the complete orthonormal system (CONS) for $L^2([0,2\pi))$ is 
$\{\mathrm{e}_\ell(\theta):=\Exp{\I \ell \theta}/\sqrt{2\pi}\}_{\ell\in\mathbb{Z}}$, 
and any element $f(\theta)$ can be expanded uniquely by this basis as 
$f(\theta)=\sum_{\ell\in\mathbb{Z}}c_\ell \mathrm{e}_\ell(\theta)$. 
We consider a Hilbert space: 
\begin{equation} \label{Hlz}
l^2(\mathbb{Z}):=\left\{\sum_{\ell\in\mathbb{Z}}c_\ell \ket{\ell} \,\Big|\, c_\ell\in\mathbb{C}, \sum_{\ell\in\mathbb{Z}}|c_\ell|^2<\infty  \right\}, 
\end{equation}
and the inner product for it by 
\begin{equation} \label{innerH}
\expectn{\psi |\phi} 
=\sum_{\ell\in\mathbb{Z}}\overline{\psi_\ell}\phi_\ell
\end{equation} 
for $\ket{\psi}=\sum_{\ell\in\mathbb{Z}}\psi_\ell\ket{\ell}, 
\ket{\phi}=\sum_{\ell\in\mathbb{Z}}\phi_\ell\ket{\ell} 
\in l^2(\mathbb{Z})$ and CONS is $\{\ket{\ell}\}_{\ell\in\mathbb{Z}}$. 
Mathematically, two Hilbert spaces are equivalent and we denote 
the Hilbert space for quantum angular momentum system by 
$\cH:=l^2(\mathbb{Z})\cong L^2([0,2\pi))$. 

To avoid a domain problem when dealing with the unbounded operator $L$, 
we study the following unitary operator generated by $L$: 
\begin{equation}
U(\theta):=\Exp{\I \theta L},
\end{equation}
which constitutes (abelian) one-parameter unitary group, 
{$\{ U(\theta)\,|\,\theta\in[0,2\pi)\}$} whose element $U(\theta)$ satisfies 
$U(\theta)^*=U(-\theta), U(\theta)U(\theta')=U(\theta+\theta')$, 
and $U(\theta+2\pi)=U(\theta)$. 
The other fundamental unitary operator is defined by 
\begin{equation}
V:=\Exp{\I\theta},
\end{equation}
which shifts the value of angular momentum by one as 
$V\Exp{\I \ell \theta}=\Exp{\I (\ell+1) \theta}$. 
These two unitaries constitute a Weyl pair and its commutation relation is given by 
\begin{equation} \label{weyl}
U(\theta)V^\ell=\Exp{\I\ell\theta}V^\ell U(\theta). 
\end{equation}

We first define an algebra generated by 
a linear hull of all Weyl pairs:
 \begin{equation}
 \alga_0:=\mathrm{span}_{\mathbb{C}}\{W(\theta,\ell) \,|\,\theta\in[0,2\pi),\ell\in\mathbb{Z} \},
\end{equation}
with 
\begin{equation}
W(\theta,\ell):=\Exp{-\I\ell\theta/2}U(\theta)V^{\ell}.   
\end{equation}
These unitary operators satisfy the following relations, 
\begin{eqnArray}\nonumber
W(\theta,\ell)^*&=&W(-\theta,-\ell),\\ \nonumber
W(\theta,\ell)W(\theta',\ell')&=&\Exp{\I(\ell'\theta-\ell\theta')/2}W(\theta+\theta',\ell+\ell'),\\
W(0,0)&=&\openone,
\end{eqnArray}
where $\openone$ denotes the identity operator on $\cH$. 
Another important property of the Weyl pair for quantum angular momentum system is 
that they are linearly independent, i.e., 
\[
\forall \theta_j\neq\theta_{j'}(j\neq j'),\quad \sum_{j,\ell} c_{\ell}(\theta_j)W(\theta_j,\ell)=0
\]
implies $c_{\ell}(\theta_j)=0$ for all $\ell$ and $\theta_j$. 

In the following, we assume an abstract algebra defined in terms of $W(\theta,\ell)$ 
and analyze algebraic properties. 
The Weyl algebra for quantum angular momentum is defined by
taking a $\sigma$-weak closure of $ \alga_0$:
 \begin{definition} \label{defamalge}(Angular momentum algebra on $\cH$)
 \begin{eqnArray}
 \alga&:=&\overline{\alga_0}^{\sigma w} \\ \nonumber
&=&\overline{ \mathrm{span}_{\mathbb{C}}\{W(\theta,\ell) \,|\,\theta\in[0,2\pi),\ell\in\mathbb{Z} \}}^{\sigma w}.
\end{eqnArray}
\end{definition}
Here the closure of an algebra $\overline{\alga}^{\sigma w} $ 
means we have included elements which are not contained 
in the original algebra $\cA$, but the limit 
under $\sigma$-weak topology (also called a physical topology \cite{araki}). 
With this additional treatment, $\alga$ becomes a von Neumann algebra 
where the double commutant theorem holds, i.e., $\cA'':=(\cA')'=\cA$. 
This subtlety is crucial to utilize many of powerful techniques in von Neumann algebras. 
To avoid being overloaded with mathematics, this technical point will not be emphasized 
unless otherwise noted. 

An important fact about the Weyl algebra (\ref{defamalge}) is 
that it coincides with the set of all bounded operators on $\cH$, 
i.e., $\alga=\algbofh$. This is shown simply by noting 
the commuing operator with all elements of $\alga$ is only 
a multiple of the identity operator, that is $\alga'=\{c\openone|c\in\mathbb{C}\}$ holds. 
Taking commutant of both sides and use $\alga''=\alga$ (von Neumann algebra), 
we get $\alga=\algbofh$. 

\section{Qubit subalgebra}\label{sec3}
We define the following operators \cite{englert06,RKSE10}: 
\begin{eqnArray}\nonumber
a_+&:=&\frac{1}{2} W(0,-1)-\frac{\I}{2} W(\pi,-1) =\frac12 (1-\Exp{\I\pi L} )V^*,\\[1ex] \nonumber
a_-&:=&\frac{1}{2} W(0,1)+\frac{\I}{2} W(\pi,1)=\frac12 (1+\Exp{\I\pi L} )V  ,\\[1ex]
a_3\, &:=&W(\pi,0)=\Exp{\I\pi L} .
\end{eqnArray}
As noted before, we have $a_1=a_+ +a_-$ and $a_2=-\I a_+ +\I a_-$. 
It is then straightforward to verify that $a_1,a_2,a_3$ and $\openone$ satisfy 
the same algebraic relation as the qubit algebra (\ref{qubitalg}). 
We define the $*$-subalgebra on $\cH$: 
\begin{definition}(Qubit subalgebra on $\cH$)\\ 
$\algqbit:=\mathrm{span}_{\mathbb{C}}\{ \openone,a_1,a_2,a_3\}$ 
\end{definition}
We can easily construct $*$-homomorphism 
$\pi: M_2(\mathbb{C})\to \algqbit$ and we can show this mapping 
is faithful, i.e., $\mathrm{ker}(\pi)=\{0\}$, to get the following result. 
\begin{lemma}
The subalgebra $\algqbit$ $*$-isomorphic to 
the qubit algebra $M_2(\mathbb{C})$, i.e., 
$\algqbit\cong M_2(\mathbb{C})$. 
\end{lemma}

To see the property of the above qubit subalgebra, we observe that 
the actions of $a_\pm,a_3$ on $\ket{\ell}\in\cH$ are 
\begin{eqnArray} \nonumber
a_+\ket{\ell}&=&
\left\{\begin{array}{r}
\ket{\ell-1}\quad(\mbox{for even $\ell$})\\
0\hspace{1.14cm}(\mbox{for odd $\ell$})
\end{array}\right.,\\[2ex] \nonumber
a_-\ket{\ell}&=&
\left\{\begin{array}{r}
0\hspace{1.14cm}(\mbox{for even $\ell$})\\
\ket{\ell+1}\quad(\mbox{for odd $\ell$})
\end{array}\right.,\\[2ex]
a_3\,\ket{\ell}&=&(-1)^\ell\ket{\ell}. 
\end{eqnArray}
Splitting the Hilbert space into a direct sum $\cH=\cH_{\mathrm{e}}\oplus\cH_{\mathrm{o}}$ with 
$\cH_{\mathrm{e}}=\left\{\sum_{\ell\in\mathbb{Z}}c_\ell \ket{2\ell} \,|\, c_\ell\in\mathbb{C}, \sum_{\ell\in\mathbb{Z}}|c_\ell|^2<\infty  \right\}$ 
and similarly for $\cH_{\mathrm{o}}$, we see that $a_+$ ($a_-$) only acts on $\cH_{\mathrm{e}}$ ($\cH_{\mathrm{o}}$) space. 

As seen from the expression of $a_3$, this unitary operator 
is not trace-class when expressed in terms CONS on $\cH$. 
This is due to the fact that the sum 
$\sum_{\ell\in\mathbb{Z}} \expectn{\ell | a_3\ell}=\sum_{\ell\in\mathbb{Z}}(-1)^\ell$ 
converges to arbitrary numbers depending upon arrangement of terms.  
One then might conclude that one cannot define a proper trace for these 
``Pauli spin operators.'' To resolve this trace problem, it is important 
to notice that different representations exist for C$^*$-algebra. 
In particular, the powerful theorem due to Gelfand-Naimark-Segal (GNS) 
provides an explicit construction of representation for a given reference state. 
Following the standard GNS construction, see Appendix A, we can show 
that our qubit subalgebra has a representation which is two-dimension 
for a given reference state as follows.

We first find an element of $\cH=l^2(\mathbb{Z})$, 
which is an eigenstate of $a_3$ with the eigenvalue $1$: 
\begin{equation} \label{psi0}
a_3 \ket{\psi}=\ket{\psi},\quad \psi\in\cH. 
\end{equation}
Requiring the normalization condition, we get 
\begin{equation}
\ket{\psi}=\sum_{\ell\in\bbz}c_{\ell} \ket{2 \ell}, 
\end{equation}
where $\sum_{\ell\in\bbz}|c_{\ell}|^2=1$.
Let us consider a state $\omega_\psi$ on $\algqbit$, i.e., 
a positive and linear function $\omega_\psi: \algqbit \to \mathbb{C}$: 
\begin{equation}
\omega_\psi(a):= \expectn{\psi\,|\,a\psi}, 
\end{equation}
for $a\in\cA_{qbit}$. 
Here, the right hand side is the inner product defined on $\cH$, that is, Eq.~(\ref{innerH}). 
The kernel of this functional is easily found as
\begin{equation}
\cK_\psi= \mathrm{span}_{\mathbb{C}}\{a_-a_+, a_+ \}. 
\end{equation}
We denote an equivalent class of $a$ by this kernel as $[a]_\psi=\{ b\in\algqbit\,|\, b-a\in\cK_\psi\}$.  
Thus, the Hilbert space associated with the GNS construction with respect to 
the reference state $\omega_\psi$, 
\begin{equation}
\cH_\psi:= \algqbit/\cK_{\psi}= \mathrm{span}_{\mathbb{C}}\{a_+a_-, a_- \},
\end{equation}
is a (quotient) subalgebra of $\algqbit$ and is two-dimensional complex Hilbert space 
with respect to the inner product: 
\begin{equation}
\expectn{a , b}_\psi:=\omega_\psi(a^*b)=  \expectn{\psi\,|\,a^*b\psi},
\end{equation}
for $a, b\in\cH_\psi$. 

It is straightforward to see that the orthonormal basis for $\cH_\psi$ is
\begin{equation} 
\{ e_0, e_1\}, \ \mathrm{with}\ e_0\in[a_+a_-]_\psi,\ e_1\in[a_-]_\psi,    
\end{equation}
so that $\expectn{e_i , e_j}_\psi=\delta_{ij}$ holds. The action of $a\in\algqbit$ 
is determined by $\pi_\psi(a_i)$ ($i=1,2,3$) as 
\begin{eqnArray}
\pi_\psi(\openone)e_0=e_0,&\ \pi_\psi(\openone)e_1=e_1,\\[1ex] 
\pi_\psi(a_1)e_0=e_1,&\ \pi_\psi(a_1)e_1=e_0,\\[1ex] 
\pi_\psi(a_2)e_0=\I e_1,&\ \pi_\psi(a_2)e_1=-\I e_0,\\[1ex] 
\pi_\psi(a_3)e_0=e_0,&\ \pi_\psi(a_3)e_1=-e_1.
\end{eqnArray}
It is clear from the above expression that the GNS construction of a cyclic 
representation for $\algqbit$ coincides with the standard one given in (\ref{pauli}). 

A trace for C$^*$-algebra $\cA$ is a linear functional satisfying 
certain axioms. Let $ \cA_+:= \{ a\in\cA\,|\, a\ge0\}$ be a subalgebra consisting of positive elements 
and $\bbr_+$ be a set of non-negative real numbers. 
A trace is a linear functional $\tau:\ \cA_+\to \bbr_+$ such that 
the following axioms hold: 
\begin{eqnArray}\nonumber
\tau(c_1 a_1+c_2 a_2)&=c_1\tau(a_1)+c_2\tau(a_2),\\
\tau(a^*a)&=\tau(aa^*),\ (\forall a\in\cA),
\end{eqnArray}
for any $c_j\in\bbr_+$ and any  $ a_j\in\cA_+$. 
It is clear that 
this definition can be extended from $\cA_+$ to the whole algebra $\cA$. 
If we define a trace for $\algqbit$ by
\begin{equation} \label{qubittr}
\mathrm{Tr}_\psi(a):=\sum_{i=0,1}\expectn{e_i, ae_i}_\psi ,
\end{equation}
for all $a\in\algqbit$, it follows that this is a well-defined trace satisfying 
the above requirement. From the definition we can also write it as
\begin{equation}
\mathrm{Tr}_\psi(a)=\expectn{\psi |a\psi}+\expectn{\psi^\bot |a\psi^\bot},
\end{equation} 
where $\ket{\psi^\bot}=\sum_{\ell} c_\ell \ket{2\ell-1}$ is normalized orthogonal state of $\ket{\psi}$. 
In the above definition, the trace depends the reference state 
and the subscript $\psi$ is indicated. 
Using this formula, we see that 
\begin{equation}
\mathrm{Tr}_\psi(\openone)=2,\quad \mathrm{Tr}_\psi(a_i)=0,\ (i=1,2,3). 
\end{equation}

We remark that the above construction works well for other elements of $\cH$. 
For example, stating with any normalized state $\ket{\psi}=\sum_\ell c_\ell\ket{\ell}\in \cH$ with 
$\sum_{\ell}|c_\ell|^2=1$, one can construct two-dimensional representation which 
is unitary equivalent to the above representation. Other possibility is to 
start with more general state on $\cH$, i.e., $\rho=\sum_{\ell,\ell'}c_{\ell\ell'}\ket{\ell}\bra{\ell'}$ 
with $c_{\ell\ell'}\in\bbc$ satisfying the conditions, $\rho\ge0$ and $\sum_{\ell\ell} c_\ell=1$, 
and define a linear functional: 
\begin{equation}
\omega_\rho(a)=\sum_{\ell\in\bbz}\expectn{\ell \,|\,\rho a\ell}. 
\end{equation} 
This construction also works for every state. We note that the notation of pure states 
for qubit subalgebra is solely defined from the properties of the linear functional $\omega$. 
It is well-known that a state $\omega$ is pure if and only if a cyclic representation $\pi_\omega$ 
is irreducible. Other equivalent conditions are also known in literature \cite{araki,Tbook,UOHbook}. 
For the algebra under consideration, a mixed state on the original Hilbert space $\cH$ can 
be a pure state for qubit subalgebra. This relative notion of purity for quantum states 
are discussed before in literature \cite{TBKN,filippo}. 

Lastly, from physical point of view, the reference state associated with the GNS construction 
is nothing but an initialization of quantum states in experimental setups. 
One then builds up any desirable quantum states by applying unitary operations to it. 
A trace defined in Eq.~(\ref{qubittr}) links from mathematical formula (algebra) 
to observable quantities, such as expectation values, probabilities, and so on. 
This link between algebraic description and experimental implementation 
shall be discussed in Sec.~\ref{implementation}. 

\section{Commutant and tensor product}\label{sec4}
\subsection{Commutant of qubit subalgebra}
The commutant of $\algqbit$ is found by solving the linear equations 
for the coefficients $c_{j,\ell}\in\mathbb{C}$; 
\begin{equation}
 \big[a_i,\,  \sum_{j,\ell:\mathrm{finite}}c_{j,\ell}W(\theta_j,\ell)\big]=0,
\end{equation} 
for $i=1,2,3$ or equivalently commutativity with $a_{\pm}$ and $a_3$. 
From commutation relation with $a_{\pm}$, 
we observe that different values of $\ell$ do not affect each other 
and it is sufficient to analyze the commutation relations for a fixed value of $\ell$, i.e., 
$\left[a_{i},\,  \sum_{j:\mathrm{finite}}c_{j}W(\theta_j,\ell)\right]=0$ for $i=1,2,3$.   
From the commutativity with $a_3$, we see that 
the only possible values for $\ell$ are even numbers, 
and $\sum_{j:\mathrm{finite}}c_{j}W(\theta_j, 2\ell)$ commutes with $a_3$. 
Observe that linear combinations $\sum_{j:\mathrm{finite}}c_{j}W(\theta_j, 2\ell)$ with the same 
value of $\ell$ can be expressed as $(\sum_{j:\mathrm{finite}}c_{j}W(\theta_j, 0))V^{2\ell}$ 
and that the operators $V^{2\ell}=W(0,2\ell)$ commute with $a_{i}(\forall i)$. 
It thus suffices to analyze $\sum_{j:\mathrm{finite}}c_{j}W(\theta_j, 0)=
\sum_{j:\mathrm{finite}}c_{j}U(\theta_j)$.

We first analyze special solution to the commutant problem. 
Consider an operator $R=\sum_{j=1,2}c_{j}(\theta_j)U(\theta_j)$ 
with $c_j(\theta_j)$ being $2\pi$ periodic functions. 
In order for $R$ to commute with $a_{\pm}$, we can show 
that two angles need to satisfy the relations $\theta_1=\theta_2+\pi$. 
Furthermore, the ratio of two coefficients are fixed as 
\begin{equation}
c_1:c_2=\Exp{-\I\theta_1/2}\cos(\theta_1/2):\I\Exp{-\I\theta_1/2}\sin(\theta_1/2).
\end{equation}
With more analysis we see that the operator $R(\theta_1)$ can 
be set as a unitary operator 
and it forms a unitary group about a parameter $\theta_1$. 
For notational convenience, we define a unitary operator by 
\begin{equation}
U_1(\theta):=\Exp{-\I\theta/2}( \cos(\theta/2)W(\theta,0)+\I\sin(\theta/2)W(\theta+\pi,0)),
\end{equation}
which satisfies $U_1(\theta)U_1(\theta')=U_1(\theta+\theta')$, 
$U_1(0)=\openone$, and $U_1(\theta+\pi)=U_1(\theta)$, that is $U_1(\theta)$ is $\pi$ periodic. 
It is concluded that $U_1(\theta)$ forms a one-parameter unitary group. 
Define the shift unitary operator
\begin{equation}
V_1:=V^2,
\end{equation}
it is shown that two unitaries satisfy the same Weyl commutation relation (\ref{weyl}) as 
\begin{equation}
U_1(\theta)V_1^\ell=\Exp{\I2\ell\theta}V_1^\ell U_1(\theta). 
\end{equation}
It is natural to define the algebra generated by these two 
unitary operators as 
\begin{equation}
\alga_1:=\overline{\mathrm{span}_{\mathbb{C}}
\{\Exp{-\I\ell 2\theta/2}U_1(\theta)V_1^{\ell} \,|\,\theta\in[0,\pi),\ell\in\mathbb{Z} \}}^{\sigma w}.
\end{equation} 

We now show that $\algqbit$ and $\alga_1$ can generate 
the original Weyl algebra. 
\begin{lemma} \label{lem2-2}
(Totality of $\algqbit$ and $\alga_1$)\\
Let $\algqbit\vee\alga_1$ be the smallest algebra containing $\algqbit$ and $\alga_1$, 
then $\algqbit\vee\alga_1=\alga$ holds.  
\end{lemma}
\begin{proof}
Since $\algqbit\vee\alga_1=\overline{\algqbit\cup\alga_1}^{\sigma w}$, 
it is enough to show that two unitaries $U(\theta)$ and $V$ can be 
generated by $\algqbit$ and $\alga_1$. 
Straightforward calculations yield 
\begin{eqnArray} \nonumber \label{inverse}
U(\theta)&=&  \Exp{-\I\theta/2}U_1(\theta)(\openone\cos(\theta/2)-a_3\sin(\theta/2)),\\[1ex] 
V&=&a_+V_1+a_-. 
\end{eqnArray}
Since $\alga_0$ is generated by the Weyl pair $\{ U(\theta),V\}$, the relation $\alga_0\subset \algqbit\cup\alga_1$ holds. 
Converse inclusion $\alga_0\supset \algqbit\cup\alga_1$ holds 
from the fact that both $\algqbit$ and $\alga_1$ are subalgebra of $\alga$. 
Taking closure of $\alga_0= \algqbit\cup\alga_1$ proves this lemma. 
$\square$
\end{proof}

Combing the above results, we have the following lemma:
\begin{lemma}\label{lem2-3}
The commutant $\algqbit'$ is generated by two unitary 
operators $U_1(\theta)$ and $V_1$ as 
\begin{eqnArray}\nonumber
\algqbit'&=&\overline{\mathrm{span}_{\mathbb{C}}\{W_1(\theta,\ell) \,|\,\theta\in[0,\pi),\ell\in\mathbb{Z} \}}^{\sigma w},\\[1ex] 
\mathrm{with}&&W_1(\theta,\ell):=\Exp{-\I\ell 2\theta/2}U_1(\theta)V_1^{\ell}.
\end{eqnArray}
\end{lemma}
\begin{proof}
It is sufficient to show $\alga_1= \algqbit'$. 
The relation $\alga_1\subset \algqbit'$ holds by construction. 
Let $\algc=\algqbit'-\alga_1$ be the set of elements in $\algqbit'$ but not in $\alga_1$. 
Consider a union of two subalgebras $\algqbit$ and $\algqbit'$, 
then $\algqbit\cup\algqbit'=(\algqbit\cup\alga_1)\cup\algc=\alga\cup\algc$ 
holds from the lemma \ref{lem2-2}. Since $\algqbit\cup\algqbit'\subset\alga$ holds 
from the definition, we have the relation $\alga\supset\alga\cup\algc$. 
Thus, $\algc=\algqbit'-\alga_1=\emptyset\Leftrightarrow \algqbit'\subset\alga_1$ 
and $\alga_1= \algqbit'$ holds. 
$\square$
\end{proof}

The following important result follows from lemma \ref{lem2-2} and $\alga_1=\algqbit'$ (Lemma \ref{lem2-3}). 
\begin{corollary} \label{cor2-5}
(Totality of $\algqbit$ and $\algqbit'$)\\
$\algqbit\vee\algqbit'=\alga$ holds.  
\end{corollary}

Another important property of commutant $\algqbit'$ is that it is $*$-isomorphic 
to the original angular momentum algebra $\alga$. This can be shown 
again by first constructing a $*$-homomorphism from 
$\alga$ to $\algqbit'$ and then showing faithfulness to obtain the following lemma.  
\begin{lemma}(Nest structure)\label{lem4-4}\\
The commutant $\algqbit'$ is $*$-isomorphic to 
the original algebra $\alga$, i.e., 
$\algqbit'\cong \alga$. 
\end{lemma}
\begin{proof}
The construction of $*$-homomorphism $\pi: \alga\to\algqbit'$ is straightforward 
by defining an angle $\theta_1=2\theta$ for $\algqbit'$. 
To show this mapping is faithful, we use the fact that 
the unitary $U_1(\theta)$ is linearly independent, 
i.e., $c_1U_1(\theta_1)+c_2U_1(\theta_2)=0$ for $\theta_1-\theta_2\neq 2\pi k$ ($k\in\mathbb{Z}$) 
implies $c_1=c_2=0$, and $V_1^{\ell}$ are also linearly independent, so as $W_1(\theta,\ell)$. 
Consider $\pi(\sum_{j,\ell:\mathrm{finite}}c_{j,\ell}W(\theta_j,\ell))
= \sum_{j,\ell:\mathrm{finite}}c_{j,\ell}W_1(\theta_j,\ell)=0$, 
this then implies all coefficients are zero, i.e., $a=0$ and $\forall a\in\alga$, $\pi(a)=0\Rightarrow a=0$ holds. 
$\square$
\end{proof}

To see the correspondence between $\alga$ and $\alga_1$, it is useful to 
define an angle $\theta_1=\theta /2$ for the algebra $\alga_1$ to make $U_1(\theta_1)$ $2\pi$ periodic rather than $\pi$. 
From the property of $U_1(\theta)$, $\{U_1(\theta_1)\,|\, \theta_1\in[0,2\pi)\}$ forms 
a unitary group. 
If the group $\{U_1(\theta_1)\,|\, \theta_1\in[0,2\pi)\}$ is (strongly) continuous, 
Stone's theorem applies and there exists a (unique) 
self-adjoint operator generating this group. We then define 
\begin{eqnArray}\nonumber
L_1&:=\displaystyle\lim_{\theta_1\to0}\frac{1}{\I\theta_1}(U_1(\theta_1)-\openone)\\[2ex] 
&=\displaystyle\frac{L}{2}-\frac{1}{4}(\openone-(-1)^{L}). 
\end{eqnArray}
Conversely, we express
\begin{equation}
U_1(\theta)=\Exp{\I2\theta L_1}=\Exp{\I\theta_1 L_1}=U_1(\theta_1). 
\end{equation}
We note this $L_1$ can be expressed using the floor function as $L_1=\floor{L/2}$. 
This expression agrees with our previous result which was obtained based on physical intuition \cite{RKSE10}. 
We however observe that this ``angular momentum'' operator $L_1$ 
needs to be carefully defined as being unbounded operator 
and it is only defined within a subset of  the Hilbert space $\cH$.

\subsection{Tensor product structure}
In this section, we shall establish our main result. 
We first show that the qubit subalgebra $\algqbit$ is a factor 
and we then construct $*$-isomorphism from $\algqbit\otimes\algqbit'$ to $\alga$ 
to conclude a tensor product structure $\alga\cong\algqbit\otimes\algqbit'$. 

\begin{lemma} \label{lem3-1}(Factor)\\
$\algqbit$ is a factor, i.e., the center is a multiple of the identity 
$\Leftrightarrow$ $Z(\algqbit):=\alga_{qbit}\cap\algqbit'=\{c\openone|c\in\mathbb{C}\}$.
\end{lemma}
\begin{proof}
Consider an element $a\in \algqbit\subset Z(\algqbit)$, which is expanded as 
$a=c_0\openone+ \sum_{j=1,2,3}c_ja_j$. 
Since this also belongs to $\algqbit'$, $c_1=c_2=0$ must hold. 
To show $c_3=0$ is equivalent to show $a_3=W(\pi,0)$ cannot be 
expanded as a linear span of $U_1(\theta)$. 
Suppose if this happens, i.e., $W(\pi,0)=\sum_{j}c_jU_1(\theta_j)$, 
then the right hand side commutes with $a_{1,2}$ since 
$U_1(\theta)\in\algqbit'$. This contradicts with the property of $\algqbit$, and $c_3=0$ holds.  
$\square$\\
\textbf{Alternative proof:}\quad 
From  corollary \ref{cor2-5}, $\algqbit\vee\algqbit'=\alga$. 
Consider commutant of both side and use $\alga'=\{c\openone|c\in\mathbb{C}\}$ 
to get $(\algqbit\vee\algqbit')'=\algqbit'\cap\algqbit''=\algqbit'\cap\algqbit=\{c\openone|c\in\mathbb{C}\}$. 
$\square$
\end{proof}

We are ready to show the main result: 
\begin{proposition} \label{prop} (Tensor product structure)\\
For the algebra and subalgebra under consideration, $\alga\cong\algqbit\otimes\algqbit'$ ($*$-isomorphism) holds. 
\end{proposition}
\begin{proof}
Consider a tensor product of two algebras, $\algqbit$ on $\cH$ 
and $\alga_1$ on $\cH$ as $\algqbit\bar{\otimes}\algqbit'$. 
(Here, the tensor product $\bar{\otimes}$ represents the one for von Neumann algebra, 
which is uniquely defined by a closure of algebraic tensor product.)
We construct a $*$-homomorphism as follows. 
Define a bilinear mapping $\pi$: $\algqbit\bar{\otimes}\algqbit'\to\alga$ as follows.  
For each $m=m_0\openone+ \sum_{j}m_ja_j\in \algqbit$ and $b\in\algqbit'$, we assign 
\begin{equation}
\pi(m\bar{\otimes}b)=m b. 
\end{equation}
Since $\algqbit$ and $\algqbit'$ commute,  the image of 
$\algqbit\bar{\otimes}\algqbit'$ by $\pi$ is the smallest algebra 
containing both subalgebras. From corollary \ref{cor2-5}, we have
\begin{equation}
\pi(\algqbit\bar{\otimes}\algqbit')=\algqbit\vee\algqbit'=\alga,
\end{equation} 
that is, this mapping is surjective.  
To show it is also injective, we note that 
$U(\theta)$ and $V$ can be represented by $\algqbit$ and $\algqbit'$ uniquely as 
Eq.~(\ref{inverse}). Thus, $\pi(a)=\pi(b)$ implies $a=b$ for $\forall a,b\in\alga$, 
and $\pi$ is one-to-one. 
Therefore, the bilinear mapping is shown to be $*$-isomorphic, 
and this proves the proposition. 
$\square$
\end{proof}

The following corollary holds immediately by combining lemma \ref{lem4-4} and the above proposition \ref{prop}. 
This proves the second result of our paper; an iterative construction of arbitrary number of qubit subalgebras. 
\begin{corollary}(Iterative construction)\\
$\alga\cong\algqbit^{\otimes n}\otimes\alga$ holds for any $n\in\mathbb{N}$.
\end{corollary}
\begin{proof}
From proposition \ref{prop}, we have $\alga\cong\algqbit\otimes\algqbit'$. 
Lemma \ref{lem4-4} gives 
\begin{equation}
\alga\cong\algqbit\otimes \cA.
\end{equation} 
By repeating the same procedure we obtain 
\begin{equation}
\alga\cong\algqbit\otimes( \algqbit\otimes\cA)
\cong\dots\cong\algqbit^{\otimes n}\otimes\alga.
\end{equation}
$\square$
\end{proof}

\section{Experimental implementation} \label{implementation}
One way to realize angular momentum algebra is to use 
the orbital angular momentum (OAM) of single photons. 
There are good review articles on OAM, see Refs.~\cite{santamato,mttt07,FAP} and references therein. 
In this section, we shall use standard symbols and conventions in physics, 
such as $\dagger$ denotes a conjugate operator and $\ket{\psi_1,\psi_2}= \ket{\psi_1}\ket{\psi_2}=
\ket{\psi_1}\otimes\ket{\psi_2}$ for a tensor product. 
A logical qubit constructed by our qubit subalgebra is called an OAM qubit for simplicity in this section. 

\subsection{Basic optical components}
The available apparatus to manipulate OAM are listed as follows. 

{\bf i) Standard linear optical tools:} \ \ A beam splitter (BS) splits the OAM into two path with a given 
transmission coefficient. A mirror flips the sign of OAM, i.e., $\ket{\ell}\to\ket{-\ell}$. 
A phase shifter (PS) introduces a phase in OAM. Other components are 
a polarizing beam splitter (PBS), a quarter wave plate (QWP), and a half wave plate (HWP).  

{\bf ii) Hologram:} \ \ A hologram can add or subtract OAM by $\Delta \ell$ with a properly designed pattern, 
that is an ideal hologram is equivalent to the unitary operator: 
\begin{equation} \label{hgU}
U_{Hg}(q)=V^{q},
\end{equation} 
with $q$ integer.  

{\bf iii) Dove prism:} \ \ A Dove prism (DP) shifts the conjugate variable of the OAM by 
a certain angle in addition to the sign change of OAM. In the following we 
omit the sign change assuming that a mirror will be accompanied with the DP. 
We denote it by the unitary operator 
\begin{equation} \label{dpU}
U_{DP}=\Exp{\I \alpha L},
\end{equation}
where $\alpha$ 
is an angle between the incoming light and the DP axis and thus take values in $[0,2\pi)$. 
DPs can be used to sort the OAM with some modulo by inserting a DP in each arm 
of the Mach-Zehnder (MZ) interferometer, see Figure \ref{fig1} (a). 
This scheme, referred to as a DP sorter, provides us accessibility to the desired OAM qubit subsytem \cite{DP1,DP2}.  

{\bf iv) Q-plate:} \ \ Lastly, a Q-plate (QP) increases and decreases the OAM depending on 
the polarization of photon \cite{QP1,QP2}. This is expressed as the unitary operator: 
\begin{equation}
U_{QP}(q)=\ket{L}\bra{R}\otimes V^{2q}+ \ket{R}\bra{L}\otimes(V^{\dagger})^{2q},
\end{equation} 
where $\ket{R} (\ket{L})$ is right (left) circularly polarized state.  Thus it acts on 
the composite Hilbert space $\cH\otimes\cH_{pol}$ with $\cH_{pol}=\bbc^2$ the polarization degrees of freedom. 
Importantly, the parameter $q$ (called a charge) can be a half integer such as $1/2$ as well as an integer. 
As is shown later of this section, the QP can convert the polarization qubit to the OAM qubit deterministically. 

It is clear that the pair $U_{Hg}$ and $U_{DP}$ form the desired Weyl pair of quantum angular momentum algebra described 
in Sec.~\ref{sec2}. Suitable combinations of two unitaries can realize 
the qubit subalgebra described in Sec.~\ref{sec3}. 

\subsection{State preparation, single qubit gates, and measurements}
There are many ways to prepare a desired OAM qubit state. For instance, Ref.~\cite{MT02} 
describes how to prepare a (infinitely many) superposition of even (odd) OAM states 
which correspond to the state (\ref{psi0}). Based on the GNS construction, 
one can then construct any states by applying suitable operators. 
As we emphasized, the coefficients of created OAM states are irrelevant when 
the qubit structure is concerned. Other possibility is to create superposition of 
different OAM states using a non-integer spiral phase plate, 
where superposition of hundreds of different OAM modes was reported \cite{O05}. 

To prepare an arbitrary single qubit state, we apply the following universal set of 
single qubit gates. 

{\bf i) Phase gate:} \ \ The single qubit phase gate can be realized by 
first sorting the OAM into even and odd modes and then shifting the phase of 
the even mode only. To perform the phase gate operation deterministically, 
we need to use another DP sortor. 
Figure \ref{fig1} (b) shows a possible implementation of the phase gate $U(\phi)={\rm diag} (1,\exp(\I\phi))$. 

{\bf ii) Hadamard gate:} \ \ The Hadamard gate $H$ requires 
$X$, $\I Y$, and $Z$ operations in between the beam splitters, which 
can be implemented by a similar setting as the phase gate. For example, 
the NOT gate $X$ can be realized in the same setting as the single qubit 
phase gate with the replacement of PS by holograms ($\Delta \ell=\pm 1$) 
inserted in each arm. The schematic experimental setting is shown in Figure \ref{fig1} (c). 

{\bf iii) Measurements:} \ \ Measurements on the single photon OAM state can be carried out by 
the DP sorter and similar manners. 
Ref.~\cite{DP1,DP2} reported OAM measurements on a single photon level 
demonstrating measurement of OAM state modulo four and eight. 
We note that there are other way to sort and then to measure OAM states with modulo $2^k$ ($k\in\bbn$) \cite{BLCBP10}. 
\begin{figure}[htbp]
\begin{center}
\includegraphics[width=12cm,keepaspectratio,clip]{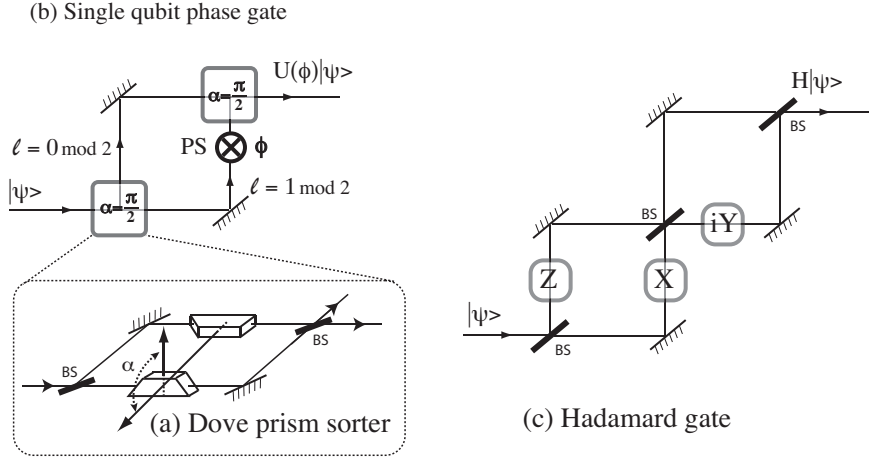}
\caption{(a) Dove prism (DP) sorter: The MZ interferometer with two DPs inserted to sort even and odd OAM 
states. (b) Single qubit phase gate. (c) The Hadmard gate using $X$, $\I Y$, and $Z$ operations.}
\label{fig1}
\end{center}
\end{figure}

\subsection{Equivalence to polarization qubit}
We briefly show that the OAM qubit is equivalent to 
other forms of photonic qubits encoded in the polarization of a single photon. 
Deterministic conversion from/to the dual rail qubit encoded 
in which path information can be shown similarly. 
The initial qubit state is encoded in the polarization as 
$\ket{\psi}_P=c_H\ket{H}+c_V\ket{V}$ where $\ket{H}$ and $\ket{V}$ are 
horizontal and vertical polarizations, respectively. The OAM qubit is assumed 
to be in the eigenstate of $a_3$:
\begin{equation}
\ket{\bar{1}}=\sum_{\ell\in\bbz}c_{\ell} \ket{2 \ell}, 
\end{equation}
where $\sum_{\ell\in\bbz}|c_{\ell}|^2=1$. 
The following orthogonal state, 
\begin{equation}
\ket{\bar{0}}=\sum_{\ell\in\bbz}c_{\ell} \ket{2 \ell-1},
\end{equation}
is expressed as $\ket{\bar{0}}=a_+\ket{\bar{1}}$ as discussed before. 
Deterministic conversion from the polarization qubit to the OAM qubit is realized by using the QP with the charge 
$q=1/2$ as shown in Figure \ref{fig2}. After passing through the quarter wave plate and 
the QP, a single photon state is sorted with a polarizing beam 
splitter depending upon the polarizations. Two arms are combined after the hologram 
and the half wave plate in each arm and pass through the MZ interferometer 
with two DPs inserted. The final state is 
\begin{equation}
\ket{\psi}_P\ket{\bar{1}}\longrightarrow\ket{H}(c_H\ket{\bar{0}}+c_V\ket{\bar{1}}).
\end{equation} 
The inverse process transfers the OAM qubit to the polarization qubit with probability one. 
\begin{figure}[htbp]
\begin{center}
\includegraphics[width=5.5cm,keepaspectratio,clip]{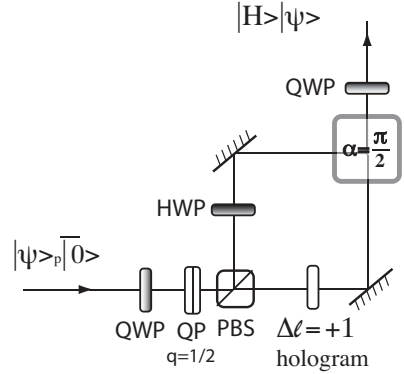}
\caption{Schematic description of deterministic conversion from 
a polarization qubit to the OAM qubit.}
\label{fig2}
\end{center}
\end{figure}

\subsection{Encoding many qubits}
It is straightforward to extend the previous discussion on the single qubit 
to multiple qubits encoded in a single OAM. The basic idea is to employ 
many MZ interferometers with DPs (the DP sorter) to sort OAM states into 
$\ell=0,1,\dots,2^n-1\; \mathrm{mod}{2^n}$. 
As an example, let us consider two OAM qubits. Figure \ref{fig3} (a) describes 
a sorting of OAM states into four arms depending on the values of angular momentum 
modulo four.  Single qubit gates on each logical qubit can be realized by 
applying the same method explained above. 

A nontrivial two qubit gate, controlled-phase gate, 
can be implemented by changing a phase of one of the four arms and then combining 
four arms into one with the aid of DP sorters as shown in Figure \ref{fig3} (b). This shows the universality 
of qubit gates for OAM qubits. 
Encoding and manipulations of many OAM qubits can be implemented similarly with more DP sorters. 
\begin{figure}[htbp]
\begin{center}
\includegraphics[width=13cm,clip]{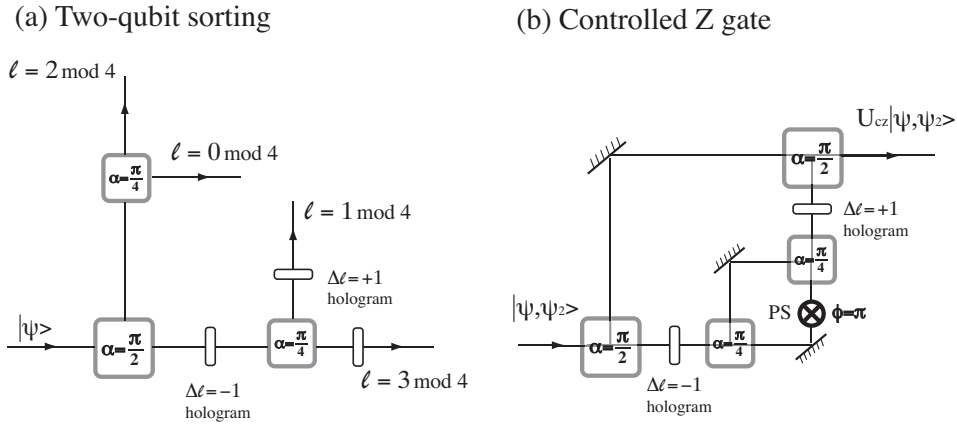}
\caption{(a) Sorting of OAM states into mod 4 states. (b) Schematic description of controlled Z gate.}
\label{fig3}
\end{center}
\end{figure}

In reality, of course, the issue of scalability needs to be addressed in order to claim for 
a realization of many qubits. The number of optical components grows quite rapidly 
as the number of logical qubits increases. The above mentioned two qubit controlled phase gate 
already demands more than six stable MZ interferometers. 
In this regard, a number of possible logical qubits is also limited by a size 
of experimental settings. A few qubits seem to be possible in the current technologies, 
and further experimental progress is needed to go more than, say, ten qubits.  

It is also possible to transfer many qubits encoded in the polarization or the dual rails 
to a single photon OAM qubits. Similar settings were studied in Ref.~\cite{GE} to transfer 
many dual rail qubits into subspaces of a single photon OAM. With the QP, it is 
possible to transfer many polarization qubits into a single photon OAM qubits as well. 
In both cases, transferring two optical qubits requires two optical CNOT gates. 
This is another advantage of OAM qubits: The ability for multiplexing transmission 
protocols with a single photon. 

\section{Extension of the proposed scheme}\label{sec6}
We briefly discuss possible extensions of our scheme and show outline for the results. 
\subsection{Qudit subalgebra}
Quantum system associated with $d$-dimensional complex Hilbert space is 
called a qudit, which can be defined as a $d\times d$ matrix algebra over 
the complex numbers $M_d(\mathbb{C})$. The same strategy applies to 
this case and one can construct a subalgebra which is $*$-isomorphic 
to $d$-dimensional matrix algebra. The key idea is to utilize the periodic projectors as follows. 

Let us consider a quantum angular momentum algebra generated by the Weyl pair $\{U(\theta),V\}$ 
and define a set of mutually orthogonal projectors: 
\begin{equation}
P^{(d)}_\ell:= \frac{1}{d}\sum_{j=0}^{d-1} \Exp{2\pi\I j (L-\ell)/d}=\frac{1}{d}\sum_{j=0}^{d-1} U\left(\frac{2\pi j}{d}\right)\Exp{-2\pi\I j\ell/d}.  
\end{equation}
Here the index $\ell$ takes $d$ different values, i.e., $\ell\in D=\{1,2,\dots, d\}$. 
Importantly, the projectors with 
different indices are orthogonal to each other and they are added to give the identity,
\begin{equation}
P^{(d)}_\ell P^{(d)}_{\ell'}=\delta_{\ell\ell'}P^{(d)}_\ell ,\quad\sum_{\ell=1}^{d}P^{(d)}_\ell=\openone . 
\end{equation}
These projectors correspond to the rank-1 diagonal elements of matrix algebra, 
which is symbolically represented by $\ket{\ell}\bra{\ell}$ ($\ell\in D$). 
This set of orthogonal projectors $P_d:=\{ P^{(d)}_j \}_{j\in D}$ can be used to define 
the desired subalgebra as follows. 

Define the following operator 
\begin{equation}
Q_{\ell\ell'}:=V^\ell P^{(d)}_d (V^*)^{\ell'}, 
\end{equation} 
and it is not difficult to show that 
it is identified with the element of the matrix algebra $\ket{\ell}\bra{\ell'}\in M_d(\mathbb{C})$. 
Thus, $Q_{\ell\ell'}$ form a basis for a $*$-subalgebra, and 
\begin{equation}
\cA_{qdit}:=\mathrm{span}_{\mathbb{C}}\{ Q_{\ell\ell'}\}_{\ell,\ell' \in D}, 
\end{equation}
constitutes the desired subalgebra which is $*$-isomorphic to $M_d(\mathbb{C})$. 

\subsection{Other infinite dimensional systems}
Another extension is to study other infinite dimensional systems. 
Our preliminary result shows similar argument applies to 
the system of quantum harmonic oscillator, where a starting Weyl algebra 
is different from the usual Weyl algebra for CCR. The fundamental operator 
is defined in terms of annihilation operator  $a$, creation operator $a^*$, 
and the number operator $N=a^* a$ as 
\begin{equation*}
W(\theta,z)=\Exp{\I \theta N/2} \Exp{ \bar{z} a-{z}a^*}\Exp{\I \theta N/2},
\end{equation*}
where $\theta\in \mathbb{R}$ and $z\in\mathbb{C}$. 
It satisfies 
\begin{eqnArray}\nonumber
W(\theta_1,z_1)W(\theta_2,z_2)&=&
\Exp{\sigma(\theta_1,z_1;\theta_2,z_2)}
W(\theta_1+\theta_2,z_1\Exp{-\I \theta_2/2}+z_2\Exp{\I \theta_1/2})\\[1ex]  
\sigma(\theta_1,z_1;\theta_2,z_2)&=&\mathrm{Im}(z_1\bar{z_2}\Exp{\I(\theta_1+\theta_2)/2})\\[1ex] 
\nonumber W(\theta,z)^*&=&W(-\theta,-z).
\end{eqnArray}

A similar construction can be realized for a qubit subalgebra out of the above algebra generated by $W(\theta,z)$. 
The main difference from the case of quantum angular momentum algebra discussed in detail is that 
we do not have unitary shift operator for the case of harmonic oscillator. Instead, we 
deal with an isometric operator which can be defined by the polar decomposition of the 
annihilation operator $a$. This additional element complicates the analysis, 
but the similar calculations hold. 
The details of harmonic oscillator case will be presented in future publication. 

\subsection{General structure}
In this section, we shall discuss the general structure of qubit subalgebra and tensor product between 
it and its commutant for a given $*$-algebra. 

Suppose a $*$-algebra $\cA$ on the Hilbert space $\cH$ is given, 
where $\cH$ is possibly infinite dimensional. 
We first construct a $*$-subalgebra $\cA_{qbit}\subset\cA$ such that 
i) $\cA_{qbit}$ contains the identity of $\cA$ and $\cA_{qbit}$ is 
$*$-isomorphic to the two-dimensional matrix algebra $M_2(\bbc)$. 
The latter condition seems to be satisfied if the original Hilbert space 
is separable and we have a CONS for it. 
The next step is to show that iii) the subalgebra and its commutant $\cA_{qbit}'$ 
span the whole algebra $\cA$, i.e., $\cA_{qbit}\vee \cA_{qbit}'=\cA$. 
As we saw in the proof of lemma \ref{lem3-1}, this condition is 
equivalent to iv) $\cA_{qbit}$ is a factor of $\cA$, i.e., 
$Z(\algqbit)=\alga_{qbit}\cap\algqbit'=\{c\openone|c\in\mathbb{C}\}$. 
This is true when the original algebra is a von Neumann algebra 
and the double commutant theorem holds. This additional assumption 
was used throughout our discussion. 

The above conditions i), ii), iii) (or i), ii), iv)) are sufficient to show that 
there exists a bilinear and faithful $*$-homomorphism from $\alga_{qbit}\otimes\alga_{qbit}'$ 
to a set of bounded operators on $\cH$ denoted as $\algbofh$ such that 
the image of $\alga_{qbit}\otimes\alga_{qbit}'$ coincides with $\alga$. 
This was proven in proposition \ref{prop} without relying on the detail structure 
of the algebra. One can also show that there are various equivalent conditions 
which purport the tensor product structure between a subalgebra which is 
$*$-isomorphic to a matrix algebra and its commutant \cite{nagaoka}.

The additional ingredient in our discussion is a simple fact that 
v) the commutant is $*$-isomorphic to the original algebra, i.e., 
$\alga_{qbit}'\cong\cA$. This is possible only if the Hilbert space is infinite dimensional. 
We have not found a simple criterion which guarantees this nest structure in terms of a given algebra. 
This will be analyzed in due course and shall be presented in elsewhere. 
If all conditions i)-v) are satisfied, then one can show the following simple construction works:
\begin{equation} \label{ntensor}
{\cal A}\cong \underbrace{{\cal A}_{qbit} \otimes{\cal A}_{qbit} \otimes\cdots\otimes{\cal A}_{qbit}}_n \otimes {\cal A}.
\end{equation}

\section{Summary}
We have shown that the Weyl algebra of quantum angular momentum can 
be decomposed into a tensor product of two algebras $\algqbit$ and $\algqbit'$. 
Here $\algqbit$ is $*$-isomorphic to two-dimensional matrix algebra, which 
we call a qubit algebra, whereas $\algqbit'$ is a commutant of $\algqbit$. 
Since the commutant is isomorphic to the original algebra, i.e., $\algqbit'\cong\alga$, 
we can iterate this construction to have arbitrary numbers of qubit subsystems. 
This paper reveals algebraic properties of these subalgebras and 
we have justified our previous result which was obtained based on physical intuition.  
We have also discussed a possible realization of these qubit subalgebras using 
orbital angular momentum of photons. In the last section, several extensions are 
outlined and this suggests that the similar constructions work for other algebras as well.

\ack
The author would like to thank Philippe Raynal, Amir Kalev, and Berge Englert 
for fruitful collaboration at early stage of this work. 
He acknowledges Prof. Hiroshi Nagaoka and Prof. Tomohiro Ogawa for 
helpful discussions on mathematical aspect of the problem. 

\appendix
\section*{Appendix. GNS construction}
\setcounter{section}{1} 
Here we give a short account on the GNS construction of representation 
for a $*$-algebra. See Ref.~\cite{araki,Tbook,UOHbook} for more details. 

For a given $*$-algebra $\cA$ and a state $\omega$, 
that is a linear functional such that 
\begin{eqnArray}\nonumber
\omega(c_1 a_1+c_2 a_2)&=&c_1\omega(a_1)+c_2\omega(a_2),\\[1ex] \nonumber
\omega(a^*a)&\ge& 0,\\[1ex] \nonumber
\omega(\openone)&=&1,
\end{eqnArray}
hold for any $c_j\in\bbc$ and any $a_j,a\in\cA$, 
we define a kernel of the linear functional $\omega$: 
\begin{equation}\nonumber
\cK_\omega:= \{a\in\cA\,|\,\omega(a^*a)=0 \}.
\end{equation}
An equivalent relation and equivalent class based on the kernel are 
introduced: 
\begin{eqnArray}\nonumber
a\sim_{\omega} b&\stackrel{{\rm def}}{\Leftrightarrow}& a-b\in\cK_\omega,\\[1ex] \nonumber
[a]_{\omega}&:=& \{ b\in\cA\,|\, a-b\in\cK_\omega \}.
\end{eqnArray}
It is straightforward to see that $\cK_\omega$ is a linear subspace of $\cA$ 
and is also a left ideal of $\cA$. 
The quotient algebra, 
\begin{eqnArray}\nonumber
\cD_\omega&:=& \cA/\cK_{\omega} \\[1ex]\nonumber 
&=&\left\{[a]_\omega \,|\,a\in\cA \right\},
\end{eqnArray}
is a subalgebra which 
consists of equivalent classes of the algebra. 

We introduce the following sesquilinear form $\cD_\omega\times\cD_\omega\to\bbc$ by
\begin{equation}\nonumber
\forall [a]_\omega,[b]_\omega\in\cD_\omega,\quad 
\expectn{[a]_\omega,[b]_\omega}_\omega:= \omega(a^*b).
\end{equation}
Owing to the equivalent relation, one can verify that 
this definition does not depend on representatives of the equivalent class and 
this defines an inner product on $\cD_\omega$. A norm 
on $\cD_\omega$ is also defined as 
\begin{equation}\nonumber
|| a ||_\omega:= \sqrt{\expectn{a,a}_\omega}=\sqrt{\omega(a^*a)}. 
\end{equation}
By completion of the inner product space $\cD_\omega$ with respect 
to the above norm, we construct a Hilbert space:
\begin{equation}\nonumber
\cH_\omega:= \overline{\cD_\omega}.
\end{equation}

For a separable $*$-algebra, the Hilbert space is also separable 
and we can find a CONS for $\cH_\omega$ 
which is denoted by $\{ [e^i]_\omega\}_{i}$. 
Since the CONS is independent of choice of representatives, 
we also denote it as $\{ e^i_\omega \}_{i}$. 

The next step is to introduce a representation $\pi_\omega: \cA\to \cH_\omega$ by
\begin{equation}\nonumber
\pi_\omega(a) [b]_\omega= [ab]_\omega, 
\end{equation}
for all $[b]_\omega\in\cH_\omega$. 
This definition is again independent of representatives and hence it 
defines a unique representation (up to unitary equivalence) on the Hilbert space $\cH_\omega$. 
Having the CONS on $\cH_\omega$, it suffices to determine 
the action of $\pi_\omega(a)$ on $e^i_\omega$, i.e., 
the coefficients $c^i_j(a)$ 
\begin{eqnArray} \nonumber
&\pi_\omega(a)e^i_\omega= [ae^i_\omega]_\omega=\sum_j c^i_j(a) e^j_\omega,  \\[1ex] \nonumber
\Leftrightarrow &\  c^i_j(a)= \expectn{e^j_\omega, [ae^i_\omega]_\omega}_\omega 
=\omega\big((e^j_\omega)^* ae^i_\omega\big),
\end{eqnArray}
determine the representation for $\cA$. This representation 
is specified by the doublet $(\cH_\omega, \pi_\omega)$ and are called as 
a GNS representation.



\section*{References}
\begin{thebibliography}{99}
\bibitem{GKP}
Gottesman D, Kitaev A, and Preskill J 2001 Encoding a qubit in an oscillator
{\it Phys.~Rev.~A} \textbf{64} 012310

\bibitem{BGS}
Bartlett S D, de Guise H, and Sanders B C 2002 
Quantum encodings in spin systems and harmonic oscillators 
{\it Phys.~Rev.~A} \textbf{65} 052316

\bibitem{RCMMG}
Ralph T C, Cilchrist A, Milburn G J, Munro W J, and Glancy S 2003 
Quantum computation with optical coherent states 
{\it Phys.~Rev.~A} \textbf{68} 042319

\bibitem{MLGWRN}
Menicucci N C, van Loock P, Gu M, Weedbrook C, Ralph T C, and Nielsen M A 2006 
Universal Quantum Computation with Continuous-Variable Cluster States 
{\it Phys.~Rev.~Lett.} {\bf 97} 110501

\bibitem{Z01}
Zanardi P 2001 
Virtual Quantum Subsystems {\it Phys.~Rev.~Lett.} {\bf 87} 077901

\bibitem{VKL}
Viola V, Knill E, and Laflamme R 2001 
Constructing qubits in physical systems {\it J.~Phys. A: Math.~Gen.} \textbf{34} 7067--7079

\bibitem{ZLL}
Zanardi P, Lidar D A, and Lloyd S 2004 
Quantum Tensor Product Structures are Observable Induced 
{\it Phys.~Rev.~Lett.} {\bf 92} 060402

\bibitem{BKOSV}
Barnum H, Knill E, Oritz G, Somma R, and Viola L 2004 
A Subsystem-Independent Generalization of Entanglement 
{\it Phys.~Rev.~Lett.} {\bf 92} 107902

\bibitem{GBRWKL}
Grace M, Brif C, Rabitz H, Walmsley I, Kosut R, and Lidar D 2006 
Encoding a qubit into multilevel subspaces 
{\it New J.~Phys.} {\bf 8} 35

\bibitem{TBKN}
Thirring W, Bertlmann R A, K\"ohler P, and Narnhofer H 2011 
Entanglement or separability: the choice of how to factorize the algebra of a density matrix 
{\it Eur.~Phys.~J.~D.} {\bf 64} 181--196

\bibitem{LS04}
Liu X F and Sun C P 2004  
On the relative quantum entanglement with respect to tensor product structure 
Preprint arXiv: quant-ph/0410245 

\bibitem{RKSE10}
Raynal P, Kalev A, Suzuki J, and B.-G. Englert B.-G 2010 
Encoding many qubits in a rotor 
{\it Phys.~Rev.~A} {\bf 81} 052327

\bibitem{apj}
Raynal P, Kalev A, Suzuki J 2014 in preparation

\bibitem{araki}
Araki H 2000 Mathematical theory of quantum fields
(Oxford University Press) 
\bibitem{Tbook}
Thirring W 2002 Quantum Mathematical Physics (Springer) 
\bibitem{UOHbook}
Umegaki H, Ohya M, and Hiai F 1985 Introduction to operator algebra 
(Kyoritsu) (in Japanese) 

\bibitem{englert06}
Englert B -G, Lee K L, Mann A, and Revzen M 2006 
Periodic and discrete Zak bases 
{\it J.~Phys. A} \textbf{39} 1669--1682

\bibitem{filippo}
de Filippo S 2000 
Quantum computation using decoherence-free states of the physical operator algebra 
{\it Phys.~Rev.~A} \textbf{62} 052307

\bibitem{santamato}
Santamato E 2004 
Photon orbital angular momentum: problems and perspectives 
{\it Fortschr.~Phys.} \textbf{52} 1141--1153 

\bibitem{mttt07}
Molina-Terriza G, Torres J P, and Torner L 2007 
Twisted photons 
{\it Nat.~Phys.} \textbf{3} 305--310

\bibitem{FAP}
Franke-Arnold S, Allen L, and Padgett M 2008 
Advances in optical angular momentum 
{\it Laser \& Photon.~Rev.} \textbf{2} 299-313

\bibitem{DP1}
Leach J, Padgett M J, Barnett S M, Franke-Arnold S, and Courtial J 2002 
Measuring the Orbital Angular Momentum of a Single Photon 
{\it Phys.~Rev.~Lett.} \textbf{88} 257901

\bibitem{DP2}
Wei H, Xue X, Leach J, Padgett M J, Barnett S M, Franke-Arnold S, Yao E, and Courtial J 2003 
Simplified measurement of the orbital angular momentum of single photons 
{\it Opt.~Commun.} \textbf{223} 117--122

\bibitem{QP1}
Marrucci L, Manzo C, and Paparo D 2006 
Optical Spin-to-Orbital Angular Momentum Conversion in Inhomogeneous Anisotropic Media 
{\it Phys.~Rev.~Lett.} {\textbf 96} 163905

\bibitem{QP2}
Nagali E, Sciarrino F, De Martini F, Marrucci L, Piccirillo B, Karimi E, and Santamato E 2009 
Quantum Information Transfer from Spin to Orbital Angular Momentum of Photons 
{\it Phys.~Rev.~Lett.} \textbf{103} 013601

\bibitem{MT02}
Molina-Terriza G, Torres J P, and Torner L 2002 
Management of the Angular Momentum of Light: Preparation of Photons in Multidimensional Vector States of Angular Momentum 
{\it Phys.~Rev.~Lett.} {\bf 88} 013601

\bibitem{O05}
Oemrawsingh S S R, Ma X, Voigt D, Aiello A, Eliel E R, 't Hooft G W, and Woerdman J P 2005 
Experimental Demonstration of Fractional Orbital Angular Momentum Entanglement of Two Photons 
{\it Phys.~Rev.~Lett.} \textbf{95} 240501

\bibitem{BLCBP10}
Berkhout G C G, Lavery M P J, Courtial J, Beijersbergen M W, and Padgett M J 2010 
Efficient Sorting of Orbital Angular Momentum States of Light 
{\it Phys.~Rev.~Lett.} {\bf 105} 153601 

\bibitem{GE}
Garc\'{\i}a-Escart\'{\i}n J C and Chamorro-Posada P 2008 
Quantum multiplexing with the orbital angular momentum of light 
{\it Phys.~Rev.~A} {\bf 78} 062320

\bibitem{nagaoka}
Nagaoka H 2013 private communication 

\end{thebibliography}
\end{document}